\documentclass[pra,floatfix,showpacs,onecolumn,nofootinbib]{revtex4-1}
\usepackage{amsfonts}
\usepackage{amsmath}
\usepackage{amssymb}
\usepackage{amsthm}
\usepackage[caption=false]{subfig}
\usepackage[mathscr]{eucal}
\usepackage[pdftex]{graphicx, color}
\usepackage{hyperref}
\usepackage{dsfont}
\newcommand{\ket}[1]{\left| #1 \right\rangle}
\newcommand{\bra}[1]{\left \langle #1 \right |}

\newcommand{\ketbra}[2]{| #1 \rangle\! \langle #2 |}
\DeclareMathOperator{\Tr}{Tr}
\DeclareMathOperator{\Ad}{Ad}
\DeclareMathOperator{\Stab}{Stab}
\DeclareMathOperator{\diag}{diag}
\DeclareMathOperator{\Trenv}{Tr_{e}}
\DeclareMathOperator{\UCent}{UCent}
\DeclareMathOperator{\Comm}{Comm}
\DeclareMathOperator{\Bicomm}{Bicomm}
\DeclareMathOperator{\Span}{Span}

\DeclareMathOperator{\LCM}{lcm}
\DeclareMathOperator{\id}{id}
\newcommand{\identity}{\mathds{1}}
\newcommand{\rmd}{\mathrm{d}}
\newcommand{\HS}{\mathcal{H}}
\newcommand{\Hs}{\HS_{\textsc{s}}}
\newcommand{\He}{\HS_{\textsc{b}}}
\newcommand{\Hse}{\Hs\otimes \He}
\newcommand{\PUH}{\mathrm{PU}(\HS)}
\newcommand{\UH}{\mathrm{U}(\HS)}

\newcommand{\Ue}{\mathrm{U}(\He)}
\newcommand{\BL}{\mathcal{B}}
\newcommand{\BH}{\BL(\HS)}
\newcommand{\BBH}{\BL\big(\BH\big)}
\newcommand{\Bs}{\BL(\Hs)}

\newcommand{\dH}{\mathfrak{d}}
\newcommand{\dHs}{\mathfrak{d}_{\textsc{s}}}
\newcommand{\dHe}{\mathfrak{d}_{\textsc{b}}}
\newtheorem{theorem}{Theorem}
\newtheorem{lemma}{Lemma}

\newtheorem{definition}{Definition}
\newtheorem{example}{Example}
\newtheorem{conjecture}{Conjecture}
\newtheorem{question}{Question}
\newtheorem{remark}{Remark}
\begin{document}
\date{ \today}
\author{Jason M. Dominy$^{(1,4)}$, Lorenzo Campos Venuti$^{(3,4)}$, Alireza Shabani$^{(5)}$, and Daniel A. Lidar$^{(1, 2,3,4)}$}
\affiliation{Departments of $^{(1)}$Chemistry, $^{(2)}$Electrical Engineering, and $^{(3)} $Physics, $^{(4)}$Center for Quantum Information Science \& Technology, University of Southern California, Los Angeles, CA 90089, USA\\
$^{(5)}$Department of Chemistry, University of California, Berkeley, CA 94720, USA}
\title{Evolution prediction from tomography}
\newcommand{\jmdstack}[2]{\genfrac{}{}{0pt}{}{#1}{#2}}

\begin{abstract}
Quantum process tomography provides a means of measuring the evolution operator for a system at a fixed measurement time $t$.  The problem of using that tomographic snapshot to predict the evolution operator at other times is generally ill-posed since there are, in general, infinitely many distinct and compatible solutions.  We describe the prediction, in some ``maximal ignorance" sense, 
of the evolution of a quantum system 
based on knowledge only of the evolution operator for finitely many times $0<\tau_{1}<\dots<\tau_{M}$ with $M\geq 1$.  To resolve the ill-posedness problem, we construct this prediction as the result of an average over some unknown (and unknowable) variables. The resulting prediction provides a description of the observer's state of knowledge of the system's evolution at times away from the measurement times. Even if the original evolution is unitary, the predicted evolution is described by a non-unitary, completely positive map.
\bigskip
\begin{quotation}
	\raggedleft
	{\em Prediction is very difficult, especially if it's about the future.}

	-- Niels Bohr
%
%
	
	{\em The future ain't what it used to be.}
	
	-- Yogi Berra
\end{quotation}

\end{abstract}
\maketitle

\section{Introduction}

Predicting the future from observation of the past is arguably one of the main motivations of physics and of science in general.  This task can be formulated as follows.  Imagine having complete knowledge of the system at some given times $\tau_{j}$, $j=1,2,\ldots,M$.  To be precise, assume that the evolution operator at such times is known.  In the quantum setting this task can be achieved using quantum process tomography \cite{Altepeter2003, Mohseni2008, Shabani2011}.  Given this information, is it possible to predict what the state of the system will be at other times? Clearly the laws of physics, such as the Hamilton or Schr\"{o}dinger equation for conservative classical or quantum mechanics, must be assumed \emph{a priori}.  In the classical case it would appear that a limitation for carrying out such a program is the exponential sensitivity on the initial conditions typical of chaotic motion.  It turns out that the main obstacle to identifying the dynamical equations from experimental data is in fact the (large) system dimensionality \cite{Cecconi2012, Cubitt2012}.  

In the quantum case, estimating the generator of the dynamics from observational data is an ill-conditioned problem, although techniques have been developed to exploit  complete positivity to alleviate this problem.  These techniques have been used to estimate the dynamics of a two-qubit system from liquid-state NMR data \cite{Boulant2003}.  However, generalizing these methods to larger Hilbert spaces turns out to be challenging, and perhaps even hopeless.  To understand the source of the difficulty, consider the case where only one snapshot is available and the evolution operator $U^{(1)}$ at time $\tau_{1}$ has been obtained.  Estimating the Hamiltonian by inverting the relation $U^{(1)}=e^{-i\tau_{1}H}$ results in infinitely many solutions corresponding to the multiple branches of the logarithm. As we will show, knowledge of the evolution operator at another time, rationally independent from $\tau_{1}$, in principle reduces the number of solutions to a single one.  However, even this approach may be infeasible.  On the one hand, solving the combinatorial problem of finding the correct branch of the logarithm for each eigenvalue has recently been shown to be NP-hard \cite{Cubitt2012}.  
On the other hand, even if this solution could be found there always exist infinitely many Hamiltonian solutions, arbitrarily far from one another, yielding evolution operators at the prescribed measurement times that are arbitrarily close to the measured evolutions.  This problem seems to become insurmountable when one realizes that $U^{(j)}$ and $\tau_{j}$ are necessarily known only to limited precision, so that the problem \emph{per se} has infinitely many solutions.  

Here we take a different approach.  Our strategy is to keep track of all these infinitely many solutions. 
It turns out that each solution can be characterized by a set of integers.  Our ignorance of the precise value of these integers reflects our lack of knowledge of the input data.  We therefore encode our ignorance of these integers into a prior distribution for the allowed evolutions.  By averaging all the allowed dynamical evolutions over this distribution, we obtain a unique quantum evolution that interpolates between the known snapshots of the dynamics. We work this out in detail in the case of closed quantum system dynamics in Section~\ref{sec:closed}, and present some preliminary results for the case of open quantum systems in Section~\ref{sec:open}. We conclude in Section~\ref{sec:conc}.

\section{Predicting Closed System Dynamics}
\label{sec:closed}

\subsection{Consistent Unitary Propagators and Admissible Hamiltonians}

We begin by considering a finite set of measurement times and a set of unitary evolution operators for a closed system, measured by quantum process tomography at these times.  We characterize the set of all Hamiltonian operators capable of generating the given evolution operators at the given times.  The solution to this problem will turn out to hinge on whether or not the times are rational ratios of one another.

\begin{definition}
	Given measurement times $0<\tau_{1}<\tau_{2}<\dots<\tau_{M}$, a corresponding set of unitary propagators $\{U^{(j)}\}\subset\UH$ will be called \emph{consistent} if there exists at least one Hamiltonian $H_{0}$ such that $U^{(j)} = e^{-\tau_{j}H_{0}}$ for all $j=1,\dots, M$.  Any Hamiltonian $H$ such that $U^{(j)} = e^{-i\tau_{j}H}$ for all $j=1,\dots, M$ will be called \emph{admissible}.
\end{definition}

	It follows that a consistent set must be mutually commutative, and therefore define a decomposition $\HS = \bigoplus_{i=1}^{\kappa}V_{i}$ into the maximal shared eigenspaces of the $\{U^{(j)}\}$, under which $U^{(j)}$ decomposes as $U^{(j)} = \bigoplus_{i=1}^{\kappa} \lambda_{i}^{(j)}\identity_{V_{i}}$.  Moreover, any admissible Hamiltonian $H$ must lie in the commutant of $\{U^{(j)}\}$ and therefore must be block diagonal with respect to this decomposition, i.e., $H = \bigoplus_{i=1}^{\kappa} H_{i}$ where $H_{i}\in\BL(V_{i})$ is Hermitian.

\begin{lemma}
	Given measurement times $0<\tau_{1}<\tau_{2}<\dots<\tau_{M}$ and a consistent set of unitary propagators $\{U^{(i)}\}$, there exists an admissible Hamiltonian $\hat{H}$ in the bicommutant (i.e., double commutant, i.e., double centralizer) of $\{U^{(i)}\}$,  $\hat{H} = \bigoplus_{i=1}^{\kappa}\hat{H}_{i} = \bigoplus_{i=1}^{\kappa}\hat{h}_{i}\identity_{V_{i}}$, where $\hat{h}_{i}\in\mathbb{R}$.
\end{lemma}
\begin{proof}
	Let $H_{0} = \bigoplus_{i=1}^{\kappa}H_{i}^{0}$ be any admissible Hamiltonian, and for each $i=1,\dots,\kappa$, let $\hat{h}_{i}\in\mathbb{R}$ be any eigenvalue of $H_{i}^{0}$.  Then the statement that $H_{0}$ is admissible is equivalent to the statement that $e^{-i\tau_{j}H_{i}^{0}} = \lambda_{i}^{(j)}\identity_{V_{i}}$ for each $i$ and $j$.  It follows that $e^{-i\tau_{j}\hat{h}_{i}} = \lambda_{i}^{(j)}$ for each $i$ and $j$.  Therefore $\hat{H} = \bigoplus_{i=1}^{\kappa}\hat{h}_{i}\identity_{V_{i}}$ is an admissible Hamiltonian in the bicommutant of $\{U^{(j)}\}$.
\end{proof}

We can now give a full characterization of the set of all admissible Hamiltonians for a set of measurement times and consistent unitary propagators.
 
\begin{lemma}
	For some $M\geq 1$, let $0<\tau_{1}<\tau_{2}<\dots<\tau_{M}$ be a set of time points and let $\{U^{(j)}\}\subset\UH$ be a corresponding consistent set of unitary propagators, with admissible Hamiltonian $\hat{H}$ in the bicommutant of $\{U^{(j)}\}$.  If the times $\{\tau_{j}\}$ are \emph{not} rationally related (at least one $\tau_{r}/\tau_{s}$ is irrational), then $\hat{H}$ is the unique possible Hamiltonian for these propagators.  If the $\{\tau_{j}\}$ are rationally related (all $\tau_{r}/\tau_{s}\in\mathbb{Q}$), then the set of admissible Hamiltonians is infinite in cardinality and given by 
	\begin{equation}
		\left\{H = \hat{H} + \frac{2\pi \LCM\{q_{j}\}}{\tau_{1}}\bigoplus_{i=1}^{\kappa}K_{i}\;:\; \text{ for all } K_{i}\in\BL(V_{i}) \text{ Hermitian with eigenvalues in }\mathbb{Z}\right\},
	\end{equation}
	where $q_{j}$ is defined by $\frac{\tau_{j}}{\tau_{1}} = \frac{p_{j}}{q_{j}}$ in normal form, and $\LCM\{q_{j}\}$ is the least common multiple of $\{q_{1},\dots,q_{M}\}$.
\label{lem:admissibleHamiltonians}
\end{lemma}
\begin{proof}
	Write $\hat{H}$ as $\hat{H} = \bigoplus_{i=1}^{\kappa}\hat{h}_{i}\identity_{V_{i}}$, and let $H=\bigoplus_{i=1}^{\kappa}H_{i}$ be any admissible Hamiltonian.  Then for each $i=1,\dots,\kappa$ and each $j=1,\dots,M$, 
	\begin{equation}
		e^{-i\tau_{j}H_{i}} = e^{-i\tau_{j}\hat{h}_{i}}\identity_{V_{i}} = \lambda_{i}^{(j)}\identity_{V_{i}}. 
	\end{equation}
	Then $e^{-i\tau_{j}(H_{i} - \hat{h}_{i}\identity_{V_{i}})} = \identity_{V_{i}}$.  So if $h_{i}$ is any eigenvalue of $H_{i}$, then $\tau_{j}(h_{i}-\hat{h}_{i})/2\pi \in\mathbb{Z}$ for all $i$ and $j$.  Therefore, either $H_{i} = \hat{h}_{i}\identity_{V_{i}}$ for all $i=1,\dots,\kappa$, or for some $i$ there exists an eigenvalue $h_{i}$ of $H_{i}$ different from $\hat{h}_{i}$ and therefore $\tau_{r}/\tau_{s} \in\mathbb{Q}$ for all $r$ and $s$.  So if the times are not rationally related, then $H = \hat{H}$ is the unique admissible Hamiltonian.  Suppose, on the other hand, that the times are rationally related. For each $j=1,\dots,M$, let $p_{j},q_{j}\in\mathbb{Z}_{+}$ be the unique relatively prime pair of positive integers such that $\tau_{j}/\tau_{1} = p_{j}/q_{j}$.  Then for each $j$, $\frac{\tau_{j}(h_{i}-\hat{h}_{i})}{2\pi}  = \frac{p_{j}}{q_{j}}\frac{\tau_{1}(h_{i}-\hat{h}_{i})}{2\pi}\in\mathbb{Z}$.   Both $\frac{\tau_{1}(h_{i}-\hat{h}_{i})}{2\pi}\in\mathbb{Z}$ and $\frac{p_{j}}{q_{j}}\frac{\tau_{1}(h_{i}-\hat{h}_{i})}{2\pi}\in\mathbb{Z}$ implies that $\frac{\tau_{1}(h_{i}-\hat{h}_{i})}{2\pi}\in q_{j}\mathbb{Z}$ for all $j$, which implies that $\frac{\tau_{1}(h_{i}-\hat{h}_{i})}{2\pi}\in q\mathbb{Z}$, where $q = \LCM\{q_{j}\}$.  Therefore $H_{i} - \hat{h}_{i}\identity_{V_{i}} = \frac{2\pi q}{\tau_{1}}K_{i}$ where $K_{i}\in\BL(V_{i})$ is Hermitian with spectrum in $\mathbb{Z}$.
\end{proof}

\subsection{Evolution Prediction Against a Prior ``Energy'' Distribution}
We now turn to the issue of deriving the predicted evolution operator 	$\Psi^{\{U^{(j)}\}}_{t}$ from a set of rationally-related measurement times $0<\tau_{1}<\tau_{2}<\dots<\tau_{M}$ and associated unitary operators $\{U^{(j)}\}$.  Under these assumptions the predicted evolution operator is an average over all of the possible unitary dynamics compatible with the tomographic data.   The result is non-unitary dynamics, specifically, a sum of a unitary evolution and projections onto the commutant and bi-commutant of the collection of measured unitary operators.  The derivation makes use of an assumption that a prior probability distribution $\mathbb{P}$ may be placed on the spectrum of the underlying Hamiltonian operator such that the eigenvalues are independent identically distributed random variables.  In practice, this distribution should represent the beliefs of user regarding the relative likelihood of dynamics of each possible speed.  The choice of prior distribution appears in the final expression as the modulus squared of the characteristic function of the distribution.  If desired, a completely uniform prior distribution may be realized as the limit of increasingly ``wide'' distributions.  The result of this ``no prior information'' assumption, however, is a predicted evolution which is discontinuous in time owing to the uniform averaging over dynamics of all speeds.

\begin{theorem}
	Fix some $M\geq 1$, let $0<\tau_{1}<\tau_{2}<\dots<\tau_{M}$ be a set of time points such that $s_{j}:=\tau_{j}/\tau_{1}$ is rational for all $j=1,\dots,M$, so that $s_{j} = p_{j}/q_{j}$ in normal form (for each $i$, $p_{i},q_{i}\in\mathbb{Z}$ are relatively prime).  Let $\{U^{(i)}\}$ be a consistent set of unitary operators representing the propagators at these $M$ time points, and let $\hat{H}\in\Bicomm\{U^{(i)}\}$ be the unique admissible Hamiltonian of minimal Hilbert-Schmidt norm with spectrum contained in $(-\pi\LCM\{q_{j}\}/\tau_{1}, \pi\LCM\{q_{j}\}/\tau_{1}]$.  Let the eigenvalues $\vec{k}\in\mathbb{Z}^{\dH}$ of $\bigoplus K_{i}$ in the description of the set of admissible Hamiltonians in Lemma \ref{lem:admissibleHamiltonians} be independent, identically distributed random variables with distribution $\mathbb{P}(\vec{k})= \prod  \mathbb{P}(k_{i})$.  Then for an arbitrary $t>0$, the map $A\mapsto\langle U(t)AU^{\dag}(t)\rangle$ is defined by averaging $e^{-itH}Ae^{itH}$ over this set of admissible Hamiltonians using this distribution, yielding the map $\Psi^{\{U^{(j)}\}}_{t}:A\mapsto \langle U(t)A U^{\dag}(t)\rangle$ for all $A\in\BH$, which is given by 
	\begin{equation}
		\Psi^{\{U^{(j)}\}}_{t}(A) = |\varphi_{\mathbb{P}}(2\pi\gamma t)|^{2}e^{-it\hat{H}}Ae^{it\hat{H}} + \big(1-|\varphi_{\mathbb{P}}(2\pi\gamma t)|^{2}\big)\Big[\Upsilon \mathcal{P}_{\{U^{(j)}\}}^{C}(A) + (\identity - \Upsilon)\mathcal{P}_{\{U^{(j)}\}}^{B}(A)\Big]
	\end{equation}
	for all $t\geq 0$, where $\varphi_{\mathbb{P}}(t)$ is the characteristic function of the distribution $\mathbb{P}$, $\mathcal{P}_{\{U^{(j)}\}}^{C}$ and $\mathcal{P}_{\{U^{(j)}\}}^{B}$ are the orthogonal projectors onto the commutant $\Comm\{U^{(j)}\}\subset \BH$ and bicommutant $\Bicomm\{U^{(j)}\}\subset \BH$ (the commutant of the commutant) of $\{U^{(j)}\}$ when $\BH$ is endowed with the Hilbert-Schmidt inner product, and where $\Upsilon\in\Bicomm\{U^{(j)}\}$ is given by $\Upsilon = \bigoplus_{i=1}^{\kappa}\frac{1}{\mu_{i}+1}\identity_{V_{i}}$ with $\mu_{i}$ the dimension of $V_{i}$.  
	\label{thm:rationalGeneralIID}
\end{theorem}
\begin{proof}
	Fix some orthonormal basis $\{|j\rangle\}$ for $\HS$ subordinate to the decomposition $\HS=\bigoplus_{i=1}^{\kappa}V_{i}$ (i.e., where each $|j\rangle$ lies in one of the subspaces $V_{i}$) and define for any $\vec{d}\in\mathbb{C}^{\dHs\dHe}$, $\diag(\vec{d}) = \sum_{j}d_{j}\ketbra{j}{j}$.  Observe that, using Lemma \ref{lem:admissibleHamiltonians}, $U(t)$ may be described by the set
	\begin{equation}
		\left\{e^{-it\left(\hat{H}+2\pi \gamma R^{\dag}\diag(\vec{k})R\right)}\;:\;\vec{k}\in\mathbb{Z}^{\dHs\dHe}, R= \bigoplus_{i=1}^{\kappa} R_{i}\text{ and } R_{i}\in\mathrm{U}(V_{i})\right\},
	\end{equation}
	where $\hat{H}$ is an admissible Hamiltonian in the bicommutant of $\{U^{(j)}\}$, and where $\gamma := \frac{\LCM\{q_{j}\}}{\tau_{1}}$.  
	Now, we introduce a probability measure as a summable, normalized function $\mathbb{P}(\vec{k})$ on the module $\mathbb{Z}^{\dHs\dHe}$ of integer vectors $\vec{k}$ parametrizing the Hamiltonian eigenvalues.  Then using Lemma \ref{lem:subspaceAdjointAverage} (see the Appendix), we can average over the set of possibilities for $U(t)$, leading to
	\begin{subequations}
	\begin{align}
		\langle U(t)AU^{\dag}(t)\rangle & = \sum_{\vec{k}\in\mathbb{Z}^{\dH}}\mathbb{P}(\vec{k})\int_{R\in\bigoplus \mathrm{U}(V_{i})}\hspace{-25pt}\rmd \eta(R) \quad e^{-it(\hat{H}+2\pi\gamma R^{\dag}\diag(\vec{k})R )}A e^{it(\hat{H}+2\pi\gamma R^{\dag}\diag(\vec{k})R)}\\
		& = \sum_{\vec{k}\in\mathbb{Z}^{\dH}}\mathbb{P}(\vec{k}) \Bigg[\left(\bigoplus_{i=1}^{\kappa} \frac{\big|\Tr\big(e^{i2\pi t\gamma \diag(\vec{k}_{(i)})}\big)\big|^{2} - 1}{\mu_{i}^{2}-1}A_{(ii)} + \frac{\mu_{i}^{2}- \big|\Tr\big(e^{i2\pi t\gamma \diag(\vec{k}_{(i)})}\big)\big|^{2}}{\mu_{i}^{2}-1}\frac{\Tr(A_{(ii)})}{\mu_{i}}\identity_{V_{i}}\right)\nonumber\\
		& \quad \oplus \bigoplus_{i\neq j=1}^{\kappa}\frac{e^{-it(\hat{h}_{i} - \hat{h}_{j})}\Tr\big(e^{-i2\pi t\gamma \diag(\vec{k}_{(i)})}\big)\Tr\big(e^{i2\pi t\gamma \diag(\vec{k}_{(j)})}\big)}{\mu_{i}\mu_{j}}A_{(ij)}\Bigg]\\
		& = \left[\bigoplus_{i=1}^{\kappa} \alpha_{i}^{P}(t) A_{(ii)} + (1-\alpha_{i}^{P}(t))\frac{\Tr(A_{(ii)})}{\mu_{i}}\identity_{V_{i}}\right] \oplus \bigoplus_{i\neq j=1}^{\kappa}\beta_{ij}(t)A_{(ij)}
	\end{align}
	\end{subequations}
	where
	\begin{subequations}
	\begin{align}
		\alpha_{i}^{P}(t) & = \sum_{\vec{k}\in\mathbb{Z}^{\dH}}\mathbb{P}(\vec{k})\frac{\big|\Tr\big(e^{i2\pi t\gamma \diag(\vec{k}^{(i)})}\big)\big|^{2} - 1}{\mu_{i}^{2}-1}\\
		\beta_{ij}(t) & = \sum_{\vec{k}\in\mathbb{Z}^{\dH}}\mathbb{P}(\vec{k})\frac{e^{-it(\hat{h}_{i} - \hat{h}_{j})}\Tr\big(e^{-i2\pi t\gamma \diag(\vec{k}^{(i)})}\big)\Tr\big(e^{i2\pi t\gamma \diag(\vec{k}^{(j)})}\big)}{\mu_{i}\mu_{j}}.
	\end{align}
	\end{subequations}
	If we take the vector components of $\vec{k}$ to be i.i.d. random variables, i.e., $\mathbb{P}(\vec{k}) = \prod_{i}\mathbb{P}(k_{i})$, then 
	\begin{subequations}
	\begin{align}
		\alpha_{i}^{P}(t) & = \sum_{\vec{k}\in\mathbb{Z}^{\mu_{i}}}\prod_{j=1}^{\mu_{i}}\mathbb{P}(k_{j})\frac{\left|\sum_{j}e^{i2\pi t\gamma k_{j}}\right|^{2} - 1}{\mu_{i}^{2}-1} \\
		& = \sum_{\vec{k}\in\mathbb{Z}^{\mu_{i}}}\prod_{j=1}^{\mu_{i}}\mathbb{P}(k_{j})\frac{\mu_{i}-1 + \sum_{j\neq l}e^{i2\pi t\gamma (k_{j}-k_{l})}}{\mu_{i}^{2}-1}\\
		& = \frac{\mu_{i}-1 + \mu_{i}(\mu_{i}-1)|\varphi_{\mathbb{P}}(2\pi\gamma t)|^{2}}{\mu_{i}^{2}-1} = \frac{\mu_{i}|\varphi_{\mathbb{P}}(2\pi\gamma t)|^{2} + 1}{\mu_{i}+1}\\
		& = |\varphi_{\mathbb{P}}(2\pi\gamma t)|^{2} + \frac{1-|\varphi_{\mathbb{P}}(2\pi\gamma t)|^{2}}{\mu_{i}+1}
	\end{align}
	\end{subequations}
	and
	\begin{subequations}
	\begin{align}
		\beta_{ij}(t) & = e^{-it(\hat{h}_{i} - \hat{h}_{j})}\sum_{\jmdstack{\vec{k}^{(i)}\in\mathbb{Z}^{\mu_{i}}}{\vec{k}^{(j)}\in\mathbb{Z}^{\mu_{j}}}}\prod_{\jmdstack{q=1,\dots,\mu_{i}}{r=1,\dots,\mu_{j}}}\mathbb{P}(k_{q}^{(i)})\mathbb{P}(k_{r}^{(j)})\frac{\left(\sum_{q=1}^{\mu_{i}}e^{-i2\pi t\gamma k_{q}^{(i)}}\right)\left(\sum_{r=1}^{\mu_{j}}e^{i2\pi t\gamma k_{r}^{(j)}}\right)}{\mu_{i}\mu_{j}}\\
		& = \frac{e^{-it(\hat{h}_{i} - \hat{h}_{j})}}{\mu_{i}\mu_{j}}\sum_{q=1}^{\mu_{i}}\sum_{r=1}^{\mu_{j}}\left(\sum_{k_{q}^{(i)}\in\mathbb{Z}}\mathbb{P}(k_{q}^{(i)})e^{-i2\pi t\gamma k_{q}^{(i)}}\right)\left(\sum_{k_{r}^{(j)}\in\mathbb{Z}}\mathbb{P}(k_{r}^{(j)})e^{i2\pi t\gamma k_{r}^{(j)}}\right)\\
		& = e^{-it(\hat{h}_{i} - \hat{h}_{j})}|\varphi_{\mathbb{P}}(2\pi\gamma t)|^{2}.
	\end{align}
	\end{subequations}
	where the $2\pi$-periodic function
	\begin{equation}
		\varphi_{\mathbb{P}}(t) = \mathbb{E}[e^{itk}] = \sum_{k\in\mathbb{Z}}\mathbb{P}(k)e^{i  t k}
	\end{equation}
	is the \emph{characteristic function} of the distribution $\mathbb{P}$.  Then 
	\begin{subequations}
	\begin{align}
		\langle U(t)AU^{\dag}(t)\rangle & = \left[\bigoplus_{i=1}^{\kappa} \alpha_{i}^{P}(t) A_{(ii)} + (1-\alpha_{i}^{P}(t))\frac{\Tr(A_{(ii)})}{\mu_{i}}\identity_{V_{i}}\right] \oplus \bigoplus_{i\neq j=1}^{\kappa}\beta_{ij}(t)A_{(ij)}\\
		& = |\varphi_{\mathbb{P}}(2\pi\gamma t)|^{2}\left[\bigoplus_{i=1}^{\kappa} A_{(ii)} \oplus \bigoplus_{i\neq j=1}^{\kappa}e^{-it(\hat{h}_{i} - \hat{h}_{j})}A_{(ij)}\right]\nonumber\\
		& \qquad  + (1-|\varphi_{\mathbb{P}}(2\pi\gamma t)|^{2})\left[\bigoplus_{i=1}^{\kappa} \frac{1}{\mu_{i}+1} A_{(ii)} + \left(1-\frac{1}{\mu_{i}+1}\right)\frac{\Tr(A_{(ii)})}{\mu_{i}}\identity_{V_{i}}\right]\\
		& = |\varphi_{\mathbb{P}}(2\pi\gamma t)|^{2}e^{-it\hat{H}}Ae^{it\hat{H}} + \big(1-|\varphi_{\mathbb{P}}(2\pi\gamma t)|^{2}\big)\Big[\Upsilon \mathcal{P}_{\{U^{(j)}\}}^{C}(A) + (\identity - \Upsilon)\mathcal{P}_{\{U^{(j)}\}}^{B}(A)\Big]
	\end{align}
	\end{subequations}
\end{proof}

\begin{example}[One qubit, $\kappa = 2$]
	Consider a closed-system problem where we have measured the projective unitary propagators of a single qubit at rationally related times $0<\tau_{1}<\dots<\tau_{M}$.  We will assume that the shared eigenspaces are 1-dimensional.  Then, letting $\hat{Z} = \identity_{V_{1}}\oplus -\identity_{V_{2}}$ (in a basis that simultaneously diagonalizes the unitaries, $\hat{Z} = \pm \sigma_{Z}$, the Pauli $Z$ matrix), we can write $\hat{H} = a \hat{Z} + b\identity$ for some $a,b\in\mathbb{R}$, and for any $A\in\BH$, 
	\begin{equation}
		\mathcal{P}_{\{U^{(j)}\}}^{C}(A) = \mathcal{P}_{\{U^{(j)}\}}^{B}(A) = \frac{1}{2}( \identity A\identity + \hat{Z} A \hat{Z}),
	\end{equation}
	so that the CP map $\Psi_{t}^{\{U^{(j)}\}}$ is a periodic interpolation between unitary evolution and a dephasing channel.  I.e.,
	\begin{subequations}
	\begin{align}
		\Psi_{t}^{\{U^{(j)}\}}(A) & = |\varphi_{\mathbb{P}}(2\pi\gamma t)|^{2}e^{-iat\hat{Z}}Ae^{iat\hat{Z}} + \frac{1}{2}(1-|\varphi_{\mathbb{P}}(2\pi\gamma t)|^{2})(\identity A\identity + \hat{Z} A \hat{Z})\\
		& =|\varphi_{\mathbb{P}}(2\pi\gamma t)|^{2}(\cos(at)\identity - i\sin(at)\hat{Z})A(\cos(at)\identity + i\sin(at)\hat{Z}) + \frac{1}{2}(1-|\varphi_{\mathbb{P}}(2\pi\gamma t)|^{2})(\identity A\identity + \hat{Z} A \hat{Z})\\
		& = |\varphi_{\mathbb{P}}(2\pi\gamma t)|^{2}(\cos^{2}(at)A - i\sin(at)\cos(at)[\hat{Z},A] + \sin^{2}(at)\hat{Z}A\hat{Z}) + \frac{1}{2}(1-|\varphi_{\mathbb{P}}(2\pi\gamma t)|^{2})(A + \hat{Z} A \hat{Z})\\
		& = \frac{1}{2}(A + \hat{Z} A \hat{Z}) + \frac{1}{2}|\varphi_{\mathbb{P}}(2\pi\gamma t)|^{2}\big(\cos(2at)(A - \hat{Z}A\hat{Z}) - i\sin(2at)[\hat{Z},A]\big)\\
		& = \frac{1}{2}(A + \hat{Z} A \hat{Z}) + \frac{1}{2}|\varphi_{\mathbb{P}}(2\pi\gamma t)|^{2}e^{-2iatZ}(A-\hat{Z}A\hat{Z}),
	\end{align}
	\end{subequations}
	so that the diagonal elements of $A$ are kept fixed, while the off-diagonal elements oscillate in magnitude according to the characteristic function $\varphi_{\mathbb{P}}$ and oscillate in phase according to the value of $a$.
\end{example}

\begin{example}[One qubit, $\kappa = 1$]
	Consider a closed-system problem where we have measured the projective unitary propagators of a single qubit at rationally related times $0<\tau_{1}<\dots<\tau_{M}$.  We will assume that there is one shared 2-dimensional eigenspace (i.e., all unitaries are just phases times identity; the measured projective unitary operators are identity for all $0<\tau_{1}<\dots<\tau_{M}$).  Then we can write $\hat{H} = a\identity$ for some $a\in\mathbb{R}$, and for any $A\in\BH$, 
	\begin{equation}
		\Upsilon \mathcal{P}_{\{U^{(j)}\}}^{C}(A) + (\identity - \Upsilon)\mathcal{P}_{\{U^{(j)}\}}^{B}(A) = \frac{1}{3}( A + \Tr(A)\identity),
	\end{equation}
	so that 
	\begin{subequations}
	\begin{align}
		\Psi_{t}^{\{U^{(j)}\}}(A) & = |\varphi_{\mathbb{P}}(2\pi\gamma t)|^{2}A + \frac{1}{3}(1-|\varphi_{\mathbb{P}}(2\pi\gamma t)|^{2})(A + \Tr(A)\identity)\\
		& = \frac{1}{3}\big(1+2|\varphi_{\mathbb{P}}(2\pi\gamma t)|^{2}\big)A + \frac{2\Tr(A)}{3}\big(1-|\varphi_{\mathbb{P}}(2\pi\gamma t)|^{2}\big)\frac{\identity}{2}
	\end{align}
	\end{subequations}
	so that this one-parameter family of CP maps describes a periodic depolarization of $A$.  Unless $\mathbb{P}(k)$ is the Kronecker delta at $k=0$, this is non-Markovian evolution.
\end{example}

\subsection{The ``Problem'' With Irrational Times}

It might be suggested that the problem of rationally related times (namely that the Hamiltonian is not uniquely defined by the unitary operators) could be resolved by having at least two time points with irrational ratio, for which there is a unique admissible Hamiltonian.  One problem with this is that any uncertainty in the time points still allows for infinitely many Hamiltonians.  Another is that, even with perfect knowledge of the time points, there exist infinitely many Hamiltonians, far from the unique perfect solution, that produce unitary operators arbitrarily close to the exact operators.

\begin{theorem}
	Fix some $M > 1$, let $0<\tau_{1}<\tau_{2}<\dots<\tau_{M}$ be a set of time points not rationally related.  Let $\{U^{(i)}\}$ be a consistent set of unitary operators representing the propagators at these $M$ time points, with unique admissible Hamiltonian $\hat{H}$.  Then for any arbitrarily small $\epsilon>0$ and any arbitrarily large $\beta>0$, there exist infinitely many Hamiltonians $H$ such that $\|H - \hat{H}\| > \beta$ and $\|e^{-i\tau_{j}H} - U^{(j)}\|<\epsilon$ for all $j=1,\dots, M$.  So, while in this irrational case the admissible Hamiltonian is uniquely defined, the inverse problem of identifying that Hamiltonian from the tomographic data is highly non-robust without additional constraints.
\end{theorem} 
\begin{proof}
	Fix an integer $m>0$ and consider a Hamiltonian of the form $H = \hat{H} + \frac{2\pi r}{\tau_{1}}K$, where $r\in\mathbb{Z}$, $0\neq K = \bigoplus_{i=1}^{\kappa} K_{i}$, and each $K_{i}\in\BL(V_{i})$ is Hermitian with integer eigenvalues in $\{-m,-m+1,\dots,m-1,m\}$.  Since $\hat{H}\in\Bicomm\{U^{(j)}\}$ and $H-\hat{H}\in\Comm\{U^{(j)}\}$, it holds that $[H-\hat{H},\hat{H}] = 0$, so that
	\begin{equation}
		e^{-i\tau_{j}H}= e^{-i\tau_{j}\hat{H}}e^{-i\tau_{j}(H-\hat{H})} = U^{(j)}\bigoplus_{i=1}^{\kappa}e^{-i2\pi r\frac{\tau_{j}}{\tau_{1}}K_{i}}
	\end{equation}
	By Dirichlet's theorem on diophantine approximation \cite{Davenport1970}, there exists an infinite sequence of positive integers $r_{k}$ and accompanying integers $\{y_{j,k}\}$ such that $\max_{1\leq j\leq M}\{|r_{k}\tau_{j}/\tau_{1} - y_{j,k}|\}\leq r_{k}^{-1/M}$.  Therefore, for any $\epsilon>0$ there exists an infinite increasing sequence of positive integers $r_{k}$ and accompanying integers $\{y_{j,k}\}$ such that $|r_{k}\tau_{j}/\tau_{1} - y_{j,k}|<\frac{\epsilon}{\dH m}$ for all $j,k$.  For any $r_{k}$ in this sequence, the spectrum of $e^{-i2\pi r_{k}\frac{\tau_{j}}{\tau_{1}}K}$ is contained in the arc $\{e^{i\theta}\;:\; \theta\in(-\epsilon/\dH,\epsilon/\dH)\}$.  Therefore for any $j=1,\dots, M$, the Hilbert-Schmidt distance between $e^{-i\tau_{j}H}$ and $U^{(j)}$ is bounded as
	\begin{equation}
		\|e^{-i\tau_{j}H} - U^{(j)}\|^{2} = \|e^{-i2\pi r_{k}\frac{\tau_{j}}{\tau_{1}}K} - \identity\|^{2}\leq\epsilon^{2}
	\end{equation}
	while $\|H - \hat{H}\| = 2\pi r_{k}\frac{\tau_{j}}{\tau_{1}}\|K\| > 2\pi r_{k}\frac{\tau_{j}}{\tau_{1}} \to \infty$ as $k\to \infty$.  So for any $\beta>0$ there will exist infinitely many $r_{k}$ such that $\|H-\hat{H}\|>\beta$. 
\end{proof}

%
%

\section{Predicting Open Quantum Dynamics}
\label{sec:open}

In order to extend the preceding analysis to open systems, we propose to first describe the set of all system-bath consistent unitary $M$-tuples that are in agreement with (i.e., would reproduce) the tomographic subsystem dynamical maps $\{\Phi_{\tau_{i}}\}$ measured at rationally-related times $0<\tau_{1}<\dots<\tau_{M}$.  Each such $M$-tuple gives rise to a maximum ignorance system-bath evolution operator $\Psi_{t}^{\{U^{(j)}\}}$.  Averaging over all of these system-bath evolution operators yields a maximum ignorance evolution prediction for the system-bath.  Application of the partial trace will yield a one-parameter family of CP maps representing the maximum ignorance prediction for the system alone.


The program outlined above begins with the identification of all system-bath unitary operators that agree with a measured CP evolution map for the system at a fixed measurement time $\tau$.  To that end, consider the evolution of a system coupled to an environment defined by the map \begin{equation}
	\rho_{\textsc{s}}(\tau) = \Trenv\big[U(\tau)(\rho_{\textsc{s}}(0)\otimes\ketbra{0}{0})U^{\dag}(\tau)\big]
\end{equation} 
for some unitary operator $U(\tau)\in\UH$, where $\HS = \Hse$.  Call this CP map $\Phi_{\tau}:\mathcal{D}_{s}\to\mathcal{D}_{s}$ on the space $\mathcal{D}_{s}$ of system density matrices.

\begin{question}
	Given the CP map $\Phi_{\tau}$ for some fixed $\tau>0$, what are the possible unitary operators $U(\tau)$ that could give rise to it?
\end{question}

\begin{conjecture}
	\label{conj:equivUnitaries}
	Let $\mathcal{A}\subset \BH$ be a unital, self-adjoint subalgebra of the associative algebra $\BH$ of bounded linear operators on $\HS = \Hse$.  Then $U,W\in\UH$ satisfy $\Trenv(UAU^{\dag}) = \Trenv(WAW^{\dag})$ for all $A\in\mathcal{A}$ (i.e., $U$ and $W$ are $\mathcal{A}$-equivalent, $U\sim_{\mathcal{A}} W$) if and only if $W = (\identity\otimes V)U Q$ for some $V\in\Ue$ and $Q\in\UCent(\mathcal{A})$, where $\UCent(\mathcal{A})$ is the ``unitary centralizer'' of $\mathcal{A}$, i.e., the set of all $R\in\UH$ such that $RAR^{\dag} = A$ for all $A\in\mathcal{A}$.
\end{conjecture}
\begin{remark}
	This version of the theorem (conjecture) covers the case $\mathcal{A} = \Bs\otimes \ketbra{0}{0}$ that we are interested in, the case $\mathcal{A} = \BH$ that was proved in \cite{Grace2010}, and many others.  Are the many other cases good for anything?
\end{remark}

\begin{lemma}
	$\UCent\big(\Bs\otimes \ketbra{0}{0}\big) \simeq  \identity_{\Hs\otimes|0\rangle}\oplus\mathrm{U}\big(\Hs\otimes (\He/\mathbb{C}|0\rangle)\big) \simeq \identity_{n_{s}}\oplus \mathrm{U}(n_{s}(n_{e}-1))$.
\end{lemma}

\begin{remark}
	That the conjecture holds for the case $\mathcal{A} = \Bs\otimes \ketbra{0}{0}$ may be seen as follows.  Fix orthonormal bases $\{|i\rangle\}\subset \Hs$ and $\{|j\rangle\}\subset \He$.  Let $E_{k} := \Trenv\big[U\big(\identity\otimes\ketbra{0}{k}\big)\big]$ be the elements of the operator sum representation (OSR).  If these matrices are known, then using the Kronecker product convention for the indexing of the basis elements of $\Hse$, the first $n_{s}$ ``columns'' of $U$ can be reconstructed as 
	\begin{equation}
		\bra{i\alpha}U\ket{j0} = \bra{i}E_{\alpha}\ket{j} 
	\end{equation}
	so that, given an OSR, the set of all unitaries $U$ yielding that exact OSR can be obtained by completing the matrix above in any way such that the result is unitary, and then multiplying on the right by any element of $\identity_{n_{s}}\oplus \mathrm{U}(n_{s}(n_{e}-1))$ to rotate the final $n_{s}(n_{e}-1)$ columns of $U$.  Finally, as pointed out in \cite{Nielsen2000} for example, there is a unitary invariance to the OSR, i.e., the OSRs $\{E_{k}\}$ and $\{F_{k}\}$ express the same quantum operation if and only if there is a $V\in\Ue$ such that $F_{k} = \sum_{j}V_{kj}E_{j}$.  In the above realization of $U$, this implies that $WU$ and $U$ express the same quantum operation if an only if $W\in\identity\otimes \Ue$.  This leaves us with the result that the only composite unitaries equivalent to $U$ are those expressed as $(\identity\otimes V) U Q$, for $V\in\Ue$ and $Q\in\identity_{n_{s}}\oplus \mathrm{U}(n_{s}(n_{e}-1))$ which is the unitary centralizer of $\mathcal{A}$.
\end{remark}

\section{Conclusions}
\label{sec:conc}

The task of reconstructing the quantum dynamics of an evolving system from the observation of finitely many experimental snapshots is technically an ill-posed problem.  Indeed, for any solution to the problem there always exist infinitely many other arbitrarily close solutions.  Mathematically, this has to do with the fact that any number can be approximated by rationals to within an arbitrarily small error.  Physically, these infinitely many solutions are a manifestation of the inevitably incomplete knowledge of the input data due to finite experimental resolution.  The problem is resolved once we postulate the existence of a prior distribution over all these admissible solutions.  We have shown, for closed quantum systems, that by averaging over all the allowed dynamical evolutions, weighted according to an unknown prior distribution, we obtain a unique non-unitary (CP) quantum evolution consistent with the observed data.  In this way our ignorance of the input data becomes explicitly encoded into the solution.  By varying the prior distribution we can interpolate from smooth solutions down to total ``pessimism'' for which the evolution is active only at the observed points. We have only touched upon the corresponding problem for open quantum systems, and plan to explore it further in a future version of this work.

\acknowledgments
This research was supported by the ARO MURI grant W911NF-11-1-0268.

\bigskip
\bibliography{EvolPrediction}

\appendix

\bigskip
\section{Schur-Weyl Duality and Analogs}

In this section we develop the group averaging results used in the proof of Theorem \ref{thm:rationalGeneralIID}.  These results may be related to the classic Schur-Weyl Duality.

\begin{lemma}[Schur-Weyl Duality] 
	Let $\mathrm{U}^{\otimes n}(\HS)=\{U^{\otimes n} = U\otimes U\otimes\cdots\otimes U\;:\; U\in\UH\}\subset \mathrm{U}(\HS^{\otimes n})$ and let $S_{n}\subset \mathrm{U}(\HS^{\otimes n})$ be the $n!$ element symmetric group represented as permutations of the $n$ subsystems, i.e., $\pi\in S_{n}$ acts as $\pi(|\psi_{1}\rangle\otimes\cdots \otimes|\psi_{n}\rangle) = |\psi_{\pi^{-1}(1)}\rangle\otimes\cdots\otimes|\psi_{\pi^{-1}(n)}\rangle$.  Then the group algebras $\mathbb{C}\mathrm{U}^{\otimes n}(\HS)$ and $\mathbb{C}S_{n}$ are each the centralizer (i.e., commutant) of the other within $\BL(\HS^{\otimes n})$ \cite{Weyl1939b, Haase1984, Doty2004, Procesi2006}.
\end{lemma}

For the results that follow, we want a very similar lemma, but on $\BH\simeq \HS\otimes\HS^{*}$, rather than on $\HS\otimes\HS$, namely:

\begin{lemma}
	Let $\PUH=\{\Ad_{U}\;:\; U\in\UH\}\subset \mathrm{U}(\BH)$ where for any $A\in\BH$ and $U\in\UH$, the adjoint action of $U$ on $A$ is given by $\Ad_{U}(A) = UAU^{\dag}$, and let $S_{2}\subset \mathrm{U}(\BH)$ be the two element group comprising the identity map $\id$ and the trace-preserving operator $Q:A\mapsto \frac{2\Tr(A)}{\dH}\identity - A$ where $\dH = \dim(\HS)$ (it is easy to check that $Q$ is a unitary involution, $Q^{*} = Q$ and $Q^{2} = \id$).  Then the group algebras $\mathbb{C}\mathcal{\PUH}$ and $\mathbb{C}S_{2}$ are each the centralizer of the other within $\BBH$, the algebra of all bounded complex-linear superoperators acting on $\BH$.
	\label{lem:SchurWeylAnalog}
\end{lemma}
\begin{proof}
	First, consider an $X\in\BBH$ that commutes with all $\Ad_{U}\in\PUH$.  Then for any $A\in \BH$, and any $\Omega\in\UH$ such that $\Omega A \Omega^{\dag} = A$, $X(A) = X(\Omega A \Omega^{\dag}) = \Omega X(A)\Omega^{\dag}$, so that $X(A)$ commutes with all unitary operators that commute with $A$.  Since $\Comm(A)$ is the complex-linear span of the unitary stabilizer $\Stab_{\UH}(A)$, this implies that $X(A)$ commutes with every operator in $\Comm(A)$, and therefore $X(A)\in\Bicomm(A)$.  Furthermore, $X\big(\ketbra{1}{1}\big)\in\Bicomm\big(\ketbra{1}{1}\big)$ implies that $X\big(\ketbra{1}{1}\big) = a\ketbra{1}{1} + b\identity$ for some coefficients $a,b\in\mathbb{C}$.  Since for any $i=1,\dots, \dH$ there exists a permutation matrix $\pi\in S_{\dH}\subset\UH$ such that $\pi|1\rangle = |i\rangle$, it follows that 
	\begin{equation}
		X(\ketbra{i}{i}) = X(\pi\ketbra{1}{1}\pi^{\dag}) = \pi X(\ketbra{1}{1})\pi^{\dag} = a\ketbra{i}{i} + b\identity.
	\end{equation}
	By complex linearity, any diagonal matrix $D$ transforms as \begin{equation}X(D) = \sum_{i}D_{ii}X(\ketbra{i}{i}) = a\sum_{i}D_{ii}\ketbra{i}{i} + b\sum_{i}D_{ii}\identity  = aD + b\Tr(D)\identity.\end{equation}Since any normal operator $A\in\BH$ is unitarily diagonalizable as $A = \Omega D \Omega^{\dag}$ for some $\Omega\in\UH$ and diagonal $D$, it follows that $X$ acts on normal operators as 
	\begin{equation}
		X(A) = X(\Omega D \Omega^{\dag}) = \Omega X(D)\Omega^{\dag} = \Omega\big[aD + b\Tr(D)\identity\big]\Omega^{\dag} = aA + b\Tr(A)\identity.
	\end{equation} Finally, since $\BH$ is spanned by the normal operators, complex linearity implies that $X$ must act as $X(A) = aA + b\Tr(A)\identity$ on all operators $A\in\BH$, so $X$ is an element of the complex span of the identity superoperator $\id$ and the map $A\mapsto \Tr(A)\identity$, which is identical to the complex span of $\id$ and the map $Q:A\mapsto\frac{2\Tr(A)}{\dH}\identity - A$.  Since $\id$ and $Q$ both trivially commute with all $\Ad_{U}\in\PUH$, defining $S_{2}:=\{\id, Q\}$, we get that $\mathbb{C}S_{2}$ is the centralizer (i.e., commutant) of $\mathbb{C}\PUH$.
	
	That $\mathbb{C}\PUH$ is the centralizer of $\mathbb{C}S_{2}$ can now be seen as a consequence of either the Schur double centralizer theorem \cite{Schur1927, Procesi2006} or the von Neumann double commutant theorem \cite{Neumann1930, Arveson1976}.
\end{proof}

\begin{remark}
	Let $Q_{\pm} := \frac{\id \pm Q}{2}$ be the projectors onto the $\pm 1$ eigenspaces of $Q$.  Then $Q_{+}(A) = \frac{\Tr(A)}{\dH}\identity$ and $Q_{-}(A) = A - \frac{\Tr(A)}{\dH}\identity$.  So the eigenspaces of $Q$ are the 1-dimensional space spanned by identity $\mathbb{C}\identity = \{a\identity\;:\;a\in\mathbb{C}\}$ and the $(\dH^{2}-1)$-dimensional space of trace zero operators, i.e., $\mathrm{sl}(\HS)$.
\end{remark}

\begin{lemma}
	For any $B\in\BH$, let $\Ad_{B}\in \BBH$ denote the adjoint operator $\Ad_{B}(A) = BAB^{\dag}$ and define the group average 
	\begin{subequations}
	\begin{equation}
		\overline{\Ad_{B}} := \int_{W\in\UH}\rmd\eta(W) \Ad_{W}\circ\Ad_{B}\circ\Ad_{W}^{*} = \int_{W\in\UH}\rmd\eta(W) \Ad_{WBW^{\dag}},
	\end{equation}
	in other words, $\overline{\Ad_{B}}$ acts on an arbitrary $A\in\BH$ as
	\begin{equation}
		\overline{\Ad_{B}}(A) = \int_{W\in\UH}\rmd\eta(W)WBW^{\dag}AWB^{\dag}W^{\dag},
	\end{equation}
	\end{subequations}
	where $\eta$ is the normalized Haar measure.   Then $\overline{\Ad_{B}}$ is given by the map 
	\begin{equation}
		\overline{\Ad_{B}}(A) = \frac{\dH|\Tr(B)|^{2} - \|B\|_{\mathrm{HS}}^{2}}{\dH(\dH^{2}-1)}A + \frac{\dH\|B\|_{\mathrm{HS}}^{2} - |\Tr(B)|^{2}}{\dH^{2}-1}\frac{\Tr(A)}{\dH}\identity
	\end{equation}
	where $\dH=\dim(\HS)$, and when $\dH=1$, this is simply $\overline{\Ad_{B}}(A) = |B|^{2}A$.  For $X\in\UH$, this becomes
	\begin{equation}
		\overline{\Ad_{X}}(A) = \frac{|\Tr(X)|^{2} - 1}{\dH^{2}-1}A + \frac{\dH^{2} - |\Tr(X)|^{2}}{\dH^{2}-1}\frac{\Tr(A)}{\dH}\identity.
	\end{equation}
	\label{lem:adjointAverage}
\end{lemma}
\begin{proof}
	Because of the invariance of the Haar measure, $\overline{\Ad_{B}}$ is readily seen to be the orthogonal projection of $\Ad_{B}$ into the centralizer of $\PUH = \{\Ad_{W}\;:\;W\in\UH\}$, which, by Lemma \ref{lem:SchurWeylAnalog}, is $\mathbb{C}S_{2} = \Span_{\mathbb{C}}\{\id, Q\}$.  In other words, 
	\begin{align}
		\overline{\Ad_{B}} & = \frac{\langle \id, \Ad_{B}\rangle}{\langle \id, \id\rangle}\id + \frac{\langle \tilde{Q}, \Ad_{B}\rangle}{\langle \tilde{Q}, \tilde{Q}\rangle }\tilde{Q}
	\end{align}
	where $\tilde{Q} := Q-\frac{\langle Q,\id\rangle}{\langle \id, \id\rangle}\id$ is the component of $Q$ orthogonal to $\id$.  It remains simply to compute the inner products, which can be computed using an orthonormal basis $\{|j\rangle\}$ for $\HS$ as $\langle X, Y\rangle = \Tr(X^{*}\circ Y) = \sum_{j,k}\big\langle \ketbra{j}{k}, X^{*}\circ Y\big(\ketbra{j}{k}\big)\big\rangle = \sum_{j,k}\big\langle X\big(\ketbra{j}{k}\big), Y\big(\ketbra{j}{k}\big)\big\rangle$.  So, first of all,
	\begin{subequations}
	\begin{align}
		\langle \id, \id\rangle & = \sum_{j,k}\big\langle \ketbra{j}{k}, \ketbra{j}{k}\big\rangle = \dH^{2}\\
		\langle Q, \id\rangle & = \sum_{j,k}\big\langle Q(\ketbra{j}{k}), \ketbra{j}{k}\big\rangle = \sum_{j\neq k}\big\langle -\ketbra{j}{k}, \ketbra{j}{k}\big\rangle + \sum_{j}\left\langle \frac{2}{\mu_{i}}\identity_{V_{i}} - \ketbra{j}{j}, \ketbra{j}{j}\right\rangle\nonumber\\
		& = \big(\dH-\dH^{2}\big) + (2-\dH) = 2-\dH^{2},
	\end{align}
	\end{subequations}
	so that $\tilde{Q}(A) = \frac{2\Tr(A)}{\dH}\identity - \frac{2}{\dH}A$.  Now,
	\begin{subequations}
	\begin{align}
		\langle \id, \Ad_{B}\rangle & = \sum_{j,k}\big\langle \ketbra{j}{k}, B\ketbra{j}{k}B^{\dag}\big\rangle = \sum_{j,k}\bra{j}B\ket{j}\bra{k}B^{\dag}\ket{k} = |\Tr(B)|^{2}\\
		\langle \tilde{Q},  \Ad_{B}\rangle & = \sum_{j,k}\big\langle \frac{2\Tr(\ketbra{j}{k})}{\dH}\identity - \frac{2}{\dH^{2}}\ketbra{j}{k}, B\ketbra{j}{k}B^{\dag}\big\rangle = \frac{2}{\dH}\sum_{j}\Tr(B\ketbra{j}{j}B^{\dag}) - \frac{2}{\dH^{2}}\sum_{j,k}\bra{j}B\ket{j}\bra{k}B^{\dag}\ket{k}\nonumber\\
		& = \frac{2}{\dH}\|B\|_{\mathrm{HS}}^{2} - \frac{2}{\dH^{2}}|\Tr(B)|^{2} = \frac{2}{\dH^{2}}\Big(\dH\|B\|_{\mathrm{HS}}^{2} - |\Tr(B)|^{2}\Big)\\
		\langle \tilde{Q},  \tilde{Q}\rangle & = \sum_{j,k}\big\langle \tilde{Q}(\ketbra{j}{k}), \tilde{Q}(\ketbra{j}{k})\big\rangle = \frac{4}{\dH^{4}}\sum_{j\neq k}\big\langle \ketbra{j}{k}, \ketbra{j}{k}\big\rangle + \sum_{j}\left\langle \frac{2}{\dH}\identity - \frac{2}{\dH^{2}}\ketbra{j}{j}, \frac{2}{\dH}\identity - \frac{2}{\dH^{2}}\ketbra{j}{j}\right\rangle\nonumber\\
		& = \frac{4(\dH-1)}{\dH^{3}} + 4 - \frac{8}{\dH^{2}} + \frac{4}{\dH^{3}} = 4\left(1-\frac{1}{\dH^{2}}\right),
	\end{align}
	\end{subequations}
	whence,
	\begin{subequations}
	\begin{align}
		\overline{\Ad_{B}} & = \frac{\langle \id, \Ad_{B}\rangle}{\langle \id, \id\rangle}\id + \frac{\langle \tilde{Q}, \Ad_{B}\rangle}{\langle \tilde{Q}, \tilde{Q}\rangle }\tilde{Q} = \frac{|\Tr(B)|^{2}}{\dH}\id + \frac{\dH\|B\|_{\mathrm{HS}}^{2} - |\Tr(B)|^{2}}{\dH^{2}-1}\tilde{Q}\\
		\overline{\Ad_{B}}(A) & = \frac{\dH|\Tr(B)|^{2} - \|B\|_{\mathrm{HS}}^{2}}{\dH(\dH^{2}-1)}A + \frac{\dH\|B\|_{\mathrm{HS}}^{2} - |\Tr(B)|^{2}}{\dH^{2}-1}\frac{\Tr(A)}{\dH}\identity.
	\end{align}
	\end{subequations}
\end{proof}

\begin{lemma}
	Let $\HS = \bigoplus_{i=1}^{\kappa}V_{i}$ be an orthogonal decomposition and $\mu_{i} = \dim(V_{i})$.  For any $B = \bigoplus_{i=1}^{\kappa} B_{i}$ with $B_{i}\in\BL(V_{i})$, define the group average
	\begin{equation}
		\overline{\Ad_{B}}^{\bigoplus}:=\int_{R\in\bigoplus_{i=1}^{\kappa}\mathrm{U}(V_{i})}\hspace{-35pt}\rmd\eta(R) \; \Ad_{R}\circ\Ad_{B}\circ\Ad_{R},
	\end{equation}
	in other words, for any $A\in\BH$, 
	\begin{equation}
		\overline{\Ad_{B}}^{\bigoplus}(A) = \int_{R\in\bigoplus_{i=1}^{\kappa}\mathrm{U}(V_{i})}\hspace{-35pt}\rmd\eta(R) \; RBR^{\dag}ARB^{\dag}R^{\dag}.
	\end{equation}
	Then 
	\begin{subequations}
	\begin{align}
		\overline{\Ad_{B}}^{\bigoplus}(A) & = \bigoplus_{i=1}^{\kappa} \frac{\mu_{i}|\Tr(B_{i})|^{2} - \|B_{i}\|_{\mathrm{HS}}^{2}}{\mu_{i}(\mu_{i}^{2}-1)}A_{(ii)} + \frac{\mu_{i}\|B\|_{\mathrm{HS}}^{2} - |\Tr(B)|^{2}}{\mu_{i}^{2}-1}\frac{\Tr(A_{(ii)})}{\mu_{i}}\identity_{V_{i}}\nonumber\\
		& \oplus \bigoplus_{i\neq j=1}^{\kappa} \frac{\Tr(B_{i})\overline{\Tr(B_{j})}}{\mu_{i}\mu_{j}}A_{(ij)}.
	\end{align}
	\end{subequations}
	\label{lem:subspaceAdjointAverage}
\end{lemma}
\begin{proof}
	First, observe that
	\begin{equation}
		\overline{\Ad_{B}}^{\bigoplus}(A) = \bigoplus_{i=1}^{\kappa}\overline{\Ad_{B_{i}}}(A_{(ii)}) \oplus\bigoplus_{i\neq j=1}^{\kappa} \int_{R_{i}\in\mathrm{U}(V_{i})}\hspace{-25pt}\rmd\eta(R_{i})\int_{R_{j}\in\mathrm{U}(V_{j})}\hspace{-25pt}\rmd\eta(R_{j}) \; R_{i}B_{i}R_{i}^{\dag}A_{(ij)}R_{j}B_{j}^{\dag}R_{j}^{\dag}
	\end{equation}
	where $\overline{\Ad_{B_{i}}}$ is as defined in Lemma \ref{lem:adjointAverage}.  The proof is then completed by invoking Lemma \ref{lem:adjointAverage} and by observing that under normalized Haar measure, 
	\begin{equation}
		\int_{R_{i}\in\mathrm{U}(V_{i})}\hspace{-25pt}\rmd\eta(R_{i}) \; R_{i}B_{i}R_{i}^{\dag} = \frac{\Tr(B)}{\mu_{i}}\identity_{V_{i}}.
	\end{equation}
\end{proof}

\bigskip
\section{Characteristic Functions of Distributions on \texorpdfstring{$\mathbb{Z}$}{Z}}

\begin{figure}[ht]
	\centering
	\subfloat{
		\centering
		\includegraphics[scale=0.48,clip=true]{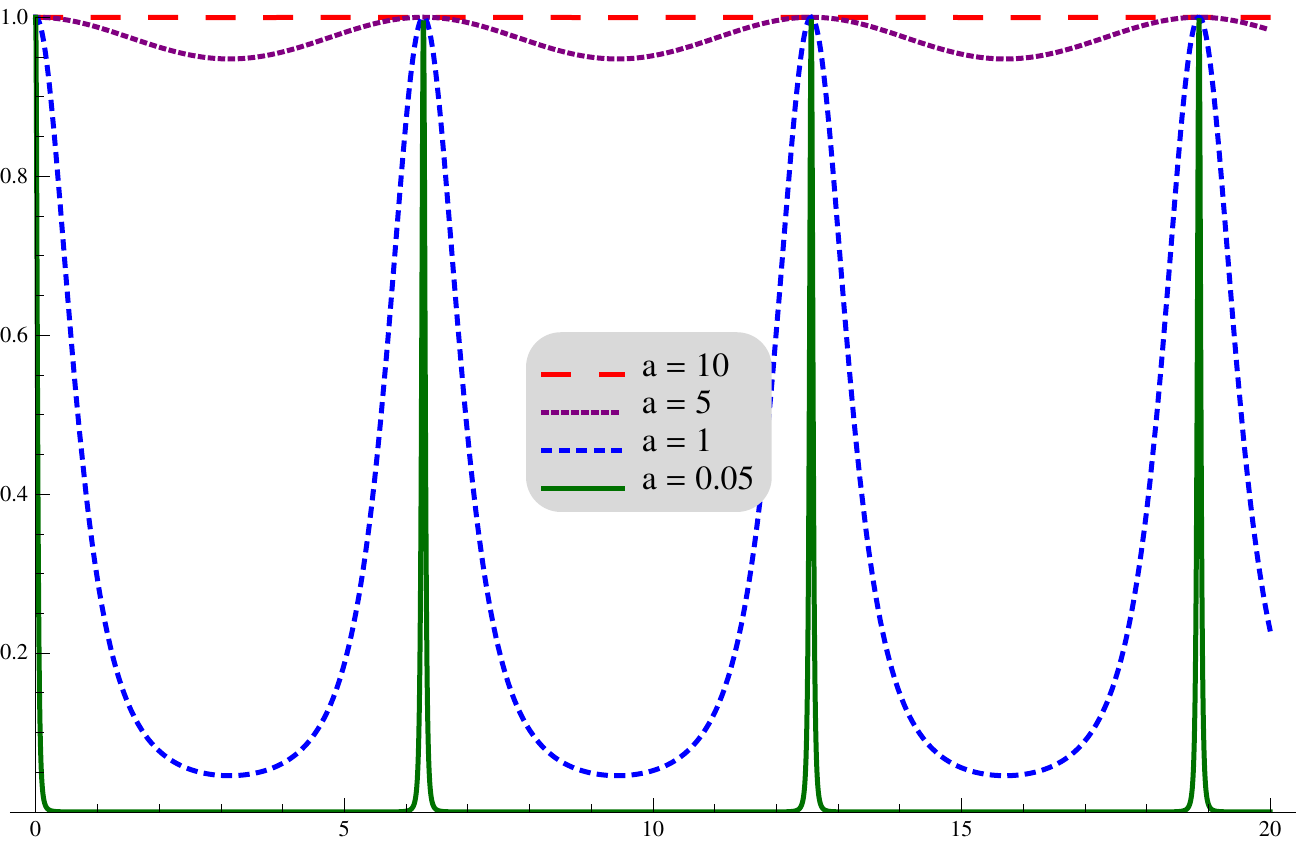}
	}
	\subfloat{
		\centering
		\includegraphics[scale=0.48,clip=true]{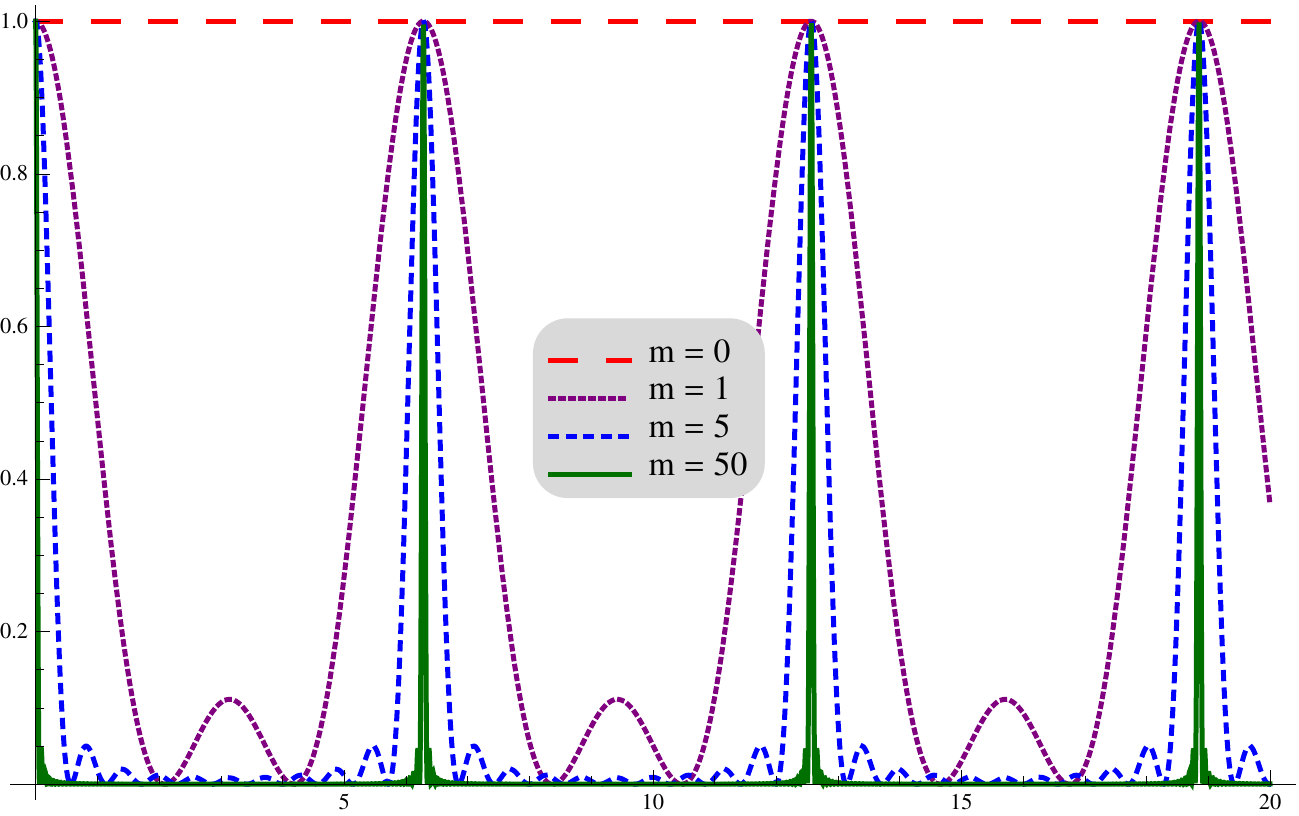}
	}
	\newline
	\subfloat{
		\centering
		\includegraphics[scale=0.48,clip=true]{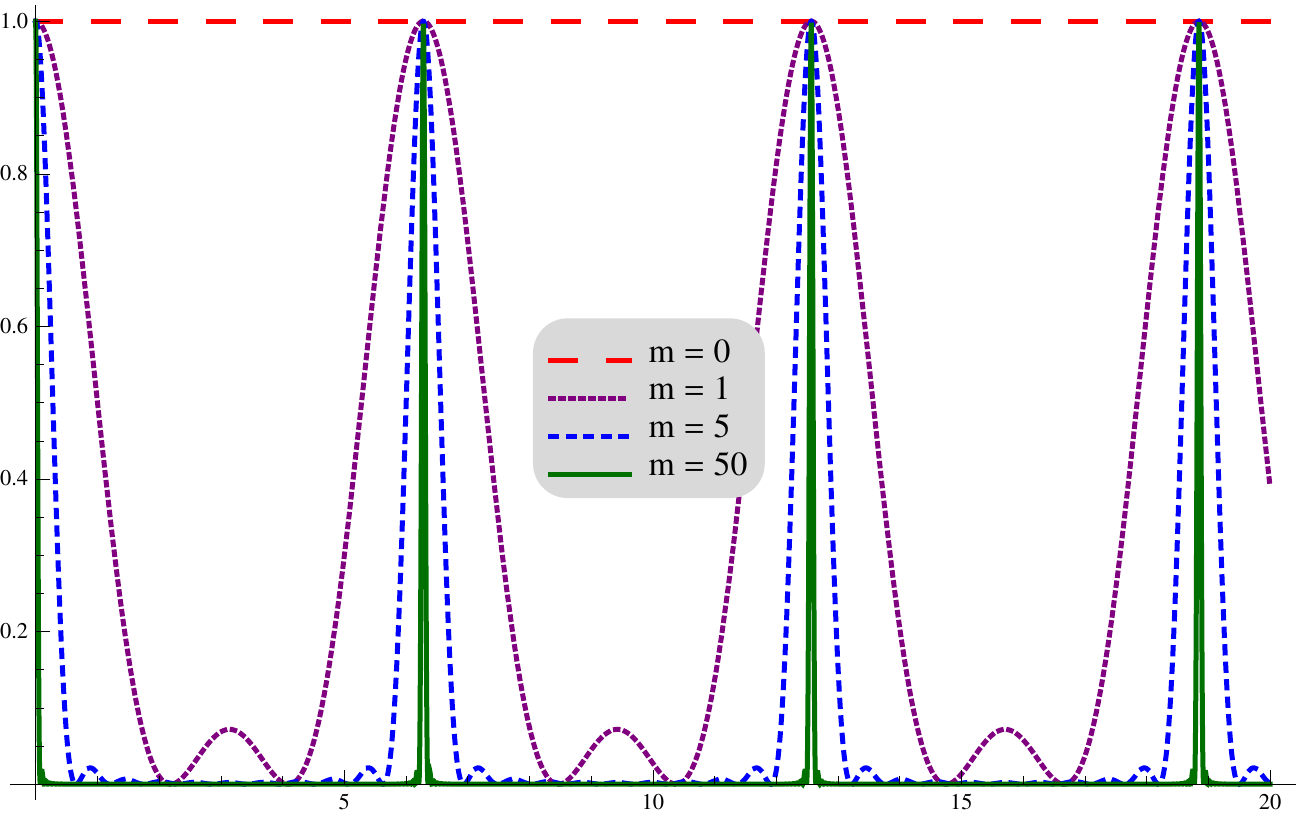}
	}
	\subfloat{
		\centering
		\includegraphics[scale=0.48,clip=true]{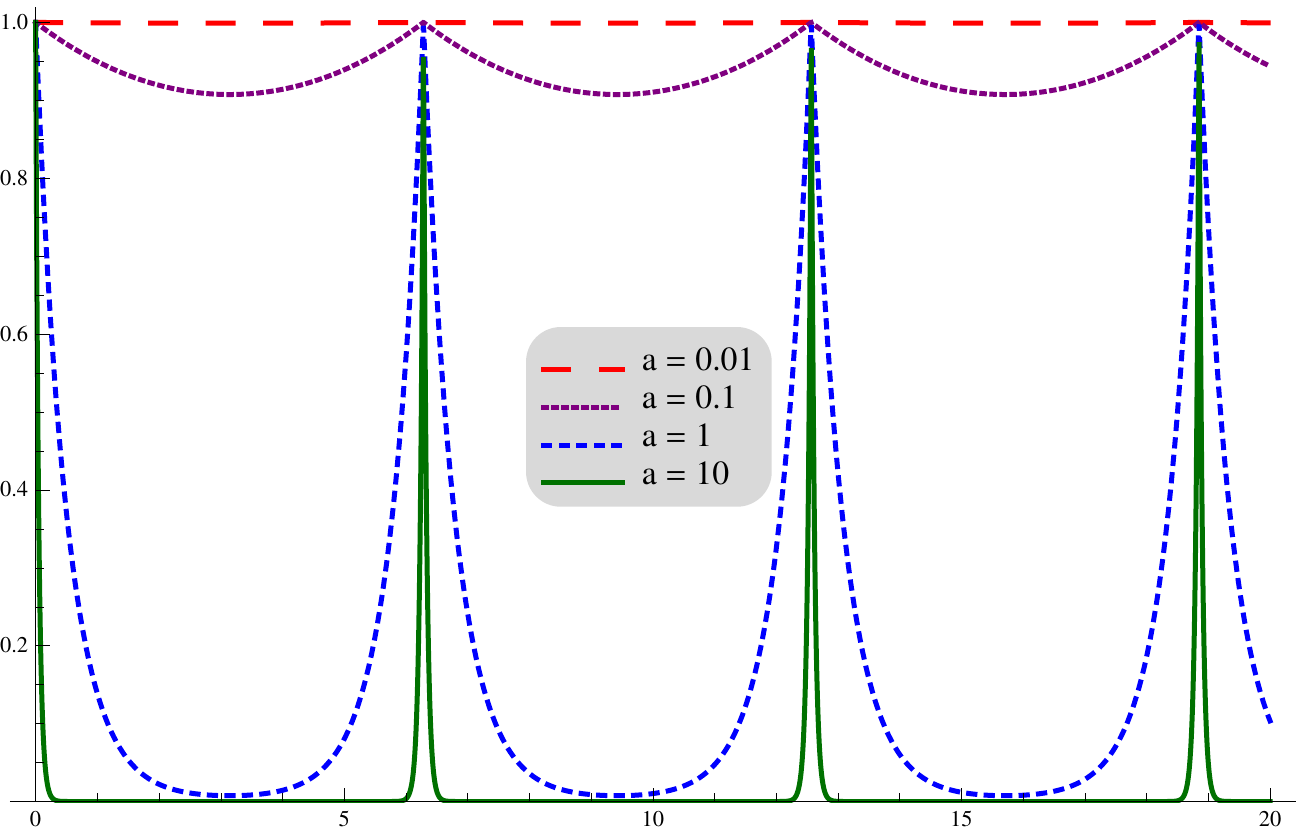}
	}
	\newline
	\subfloat{
		\centering
		\includegraphics[scale=0.48,clip=true]{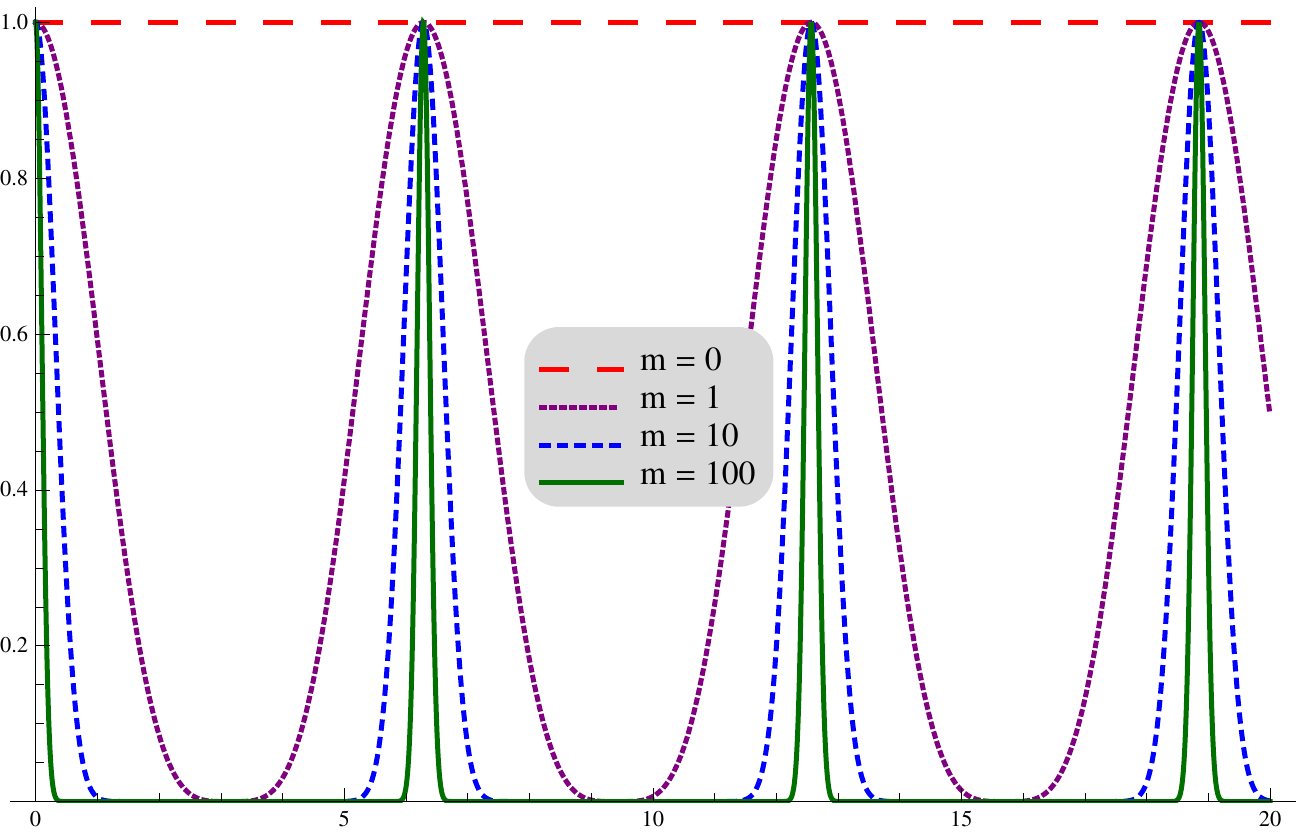}
	}
	\subfloat{
		\centering
		\includegraphics[scale=0.48,clip=true]{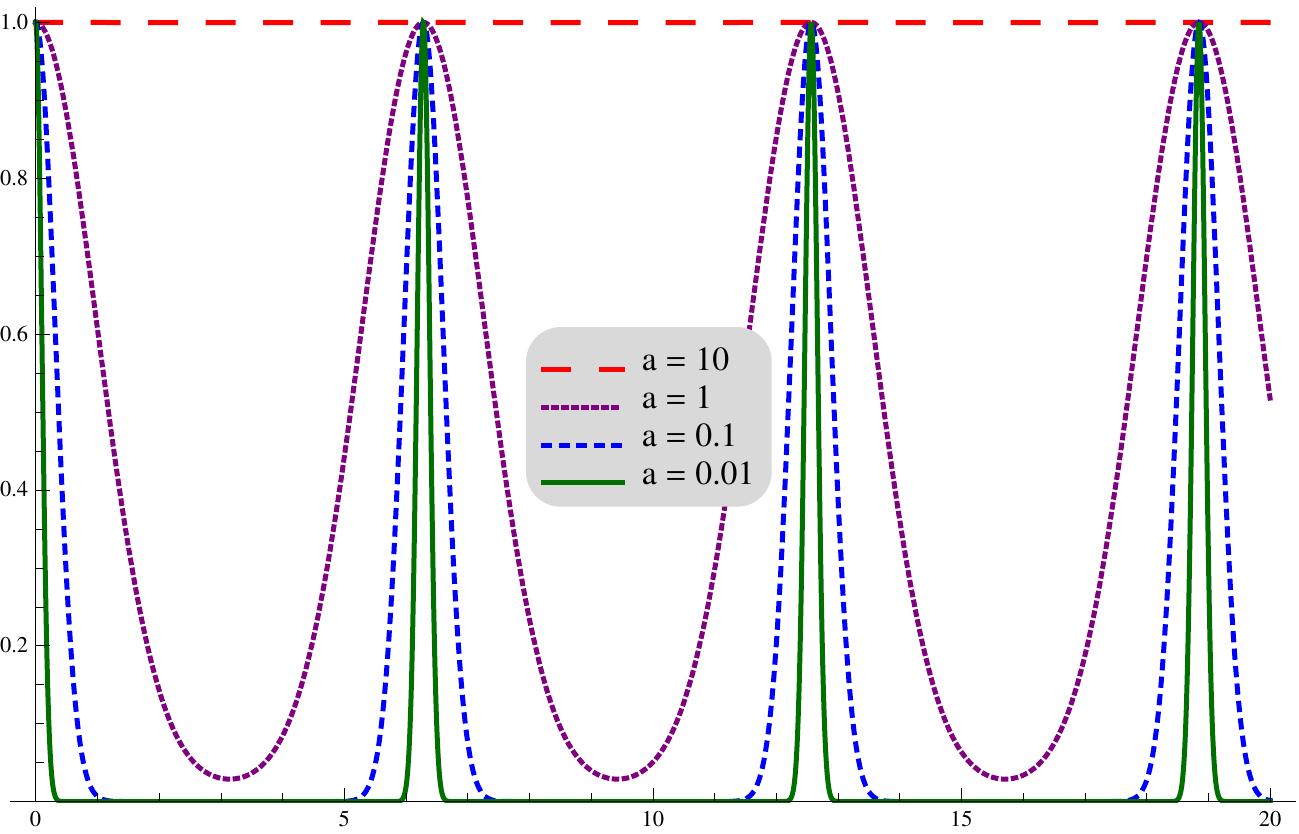}
	}
	\caption{Plots of the squared characteristic functions $|\varphi_{\mathbb{P}}(t)|^{2}$ for several probability distributions $\mathbb{P}(k)$ on $\mathbb{Z}$.  In left-to-right, top-to-bottom order, they are the characteristic functions for: (a) the exponential distribution, (b) the truncated uniform distribution, (c) the semi-circular distribution, (d) the Cauchy-Lorentz distribution, (e) the binomial distribution, and (f) the normal distribution.}
	\label{fig:squaredCharFunctions}
\end{figure}
We take a look at the characteristic functions for several common distributions on $\mathbb{Z}$.  The squared characteristic functions are plotted in Figure \ref{fig:squaredCharFunctions}.  All of the distribution families in Table \ref{tab:distChar} span the range from the Kronecker delta at $k=0$ (the characteristic function of which is the constant function $\varphi_{\mathbb{P}}(t) = 1$) to the untruncated uniform distribution on $\mathbb{Z}$ (with characteristic function the discontinuous function which is zero everywhere except at integer multiples of $2\pi$, where it takes the value 1).  As all of these distributions are symmetric unimodal distributions centered at $k=0$, they all have broadly similar characteristic functions.  It may be noticed, however, that the characteristic function of the Cauchy-Lorentz distribution has cusps at integer multiples of $2\pi$, which is related to the fact that the higher even moments $\mathbb{E}[k^{2q}] = \sum_{k}k^{2q}\mathbb{P}(k)$ do not converge for $q\geq 1$.  
\begin{table}[ht]
	\caption{Some common probability distributions on $\mathbb{Z}$ and the associated characteristic functions.}
	\begin{tabular}{l|l|l}
		Distribution & $\mathbb{P}(k)$ & $\varphi_{\mathbb{P}}(t)$\\
		\hline
		Exponential & $\mathbb{P}(k) = \frac{\cosh(a)-1}{\sinh(a)}e^{-a|k|}$ & $\varphi_{\mathbb{P}}(t) = \frac{\cosh(a)-1}{\cosh(a)-\cos(t)}$\\
		Truncated Uniform & $\mathbb{P}(k) = \begin{cases}\frac{1}{2m+1} & |k|\leq m\\ 0 & \text{else}\end{cases}$ & $\varphi_{\mathbb{P}}(t) = \frac{\sin\left[\left(\frac{2m+1}{2}\right)t\right]}{(2m+1)\sin(t/2)}$\\
		Semicircular & $\mathbb{P}(k) = \begin{cases}\frac{\sqrt{(m+1)^{2}-k^{2}}}{\sum_{l=-m}^{m}\sqrt{(m+1)^{2}-l^{2}}} & |k|\leq m\\ 0 & \text{else}\end{cases}$ & $\varphi_{\mathbb{P}}(t) = \frac{\sum_{k=-m}^{m}e^{ikt}\sqrt{(m+1)^{2}-k^{2}}}{\sum_{k=-m}^{m}\sqrt{(m+1)^{2}-k^{2}}}$\\
		Cauchy-Lorentz & $\mathbb{P}(k) = \frac{\tanh(a\pi)}{\pi}\frac{a}{a^{2}+k^{2}}$ & $\varphi_{\mathbb{P}}(t) = \frac{\tanh(a\pi)}{\pi}\sum_{k=-\infty}^{\infty}\frac{a e^{ikt}}{a^{2}+k^{2}}$\\
		Binomial & $\mathbb{P}(k) = \begin{cases}2^{-2m}\binom{2m}{k+m} & |k|\leq m\\0 & \text{else}\end{cases}$ & $\varphi_{\mathbb{P}}(t) = 2^{-m}(1+\cos(t))^{m}$\\
		Normal & $\mathbb{P}(k) = \frac{e^{-a k^{2}}}{\sum_{l=-\infty}^{\infty}e^{-a l^{2}}}$ & $\varphi_{\mathbb{P}}(t) = \frac{\sum_{k=-\infty}^{\infty}e^{-a k^{2}}e^{ikt}}{\sum_{k=-\infty}^{\infty}e^{-a k^{2}}}$
	\end{tabular}
	\label{tab:distChar}
\end{table}

\end{document}